\documentclass[12pt]{article}

\usepackage{amsfonts}
\usepackage{amssymb}
\usepackage{amsthm}
\usepackage{amsmath}
\usepackage{MnSymbol}

\usepackage{graphicx}
\usepackage{setspace}
\usepackage[margin=1in]{geometry}
\usepackage{multicol}
\usepackage[english]{babel}
\usepackage{blindtext}
\usepackage{fancyhdr,lastpage}

\usepackage{color}

\newtheorem{thm}{Theorem}[section]
\newtheorem{prop}[thm]{Proposition}
\newtheorem{cor}[thm]{Corollary}
\newtheorem{lem}[thm]{Lemma}
\newtheorem{fact}[thm]{Fact}

{\theoremstyle{definition} \newtheorem{defn}[thm]{Definition}}

\newcommand{\Loop}{\text{Loop}}
\newcommand{\skel}{\text{skel}}
\newcommand{\Bary}{\text{Bary}}
\newcommand{\ld}{\lambda}
\newcommand{\K}{\mathcal{K}}
\newcommand{\Ll}{\mathcal{L}}
\newcommand{\N}{\mathbb{N}}
\newcommand{\Prt}{\mathcal{P}}

\newcommand{\I}{\mathcal{I}}
\newcommand{\Oh}{\mathcal{O}}
\newcommand{\dl}{\delta}
\newcommand{\id}{\text{id}}
\newcommand{\C}{\mathcal{C}}

\newcommand{\Z}{\mathbb{Z}}

\newcommand{\Ob}{\text{Ob}}
\newcommand{\Hom}{\text{Hom}}

\newcommand{\A}{\mathcal{A}}
\newcommand{\B}{\mathcal{B}}

\begin{document}

\title{The Relative Power of Composite Loop Agreement Tasks}

\author{Maurice Herlihy and Vikram Saraph}

\maketitle

\begin{abstract}

Loop agreement is a family of wait-free tasks that includes set agreement and simplex agreement, and was used to prove the undecidability of wait-free solvability of distributed tasks by read/write memory. Herlihy and Rajsbaum defined the algebraic signature of a loop agreement task, which consists of a group and a distinguished element. They used the algebraic signature to characterize the relative power of loop agreement tasks. In particular, they showed that one task implements another exactly when there is a homomorphism between their respective signatures sending one distinguished element to the other. In this paper, we extend the previous result by defining the \emph{composition} of multiple loop agreement tasks to create a new one with the same combined power. We generalize the original algebraic characterization of relative power to compositions of tasks. In this way, we can think of loop agreement tasks in terms of their basic building blocks. We also investigate a category-theoretic perspective of loop agreement by defining a category of loops, showing that the algebraic signature is a functor, and proving that our definition of task composition is the ``correct" one, in a categorical sense.

\end{abstract}


\section{Introduction}

A \emph{task} is a distributed problem in which each process begins with an input, communicates with others, and returns an output according to the task's specification. Common examples of tasks include consensus \cite{Consensus}, set agreement \cite{SetAgree}, and renaming \cite{Rename}. \emph{Protocols} are distributed programs that solve tasks. A protocol is \emph{wait-free} if every non-faulty process running the protocol eventually finishes execution, regardless of other process failures. One task \emph{implements} another if a protocol for the first task can be modified in a simple way to solve the second task.

\emph{Loop agreement} is a family of tasks that models the convergence of processes along a distinguished loop of a given space, and includes simplex agreement and set agreement. One application of loop agreement is a simple proof of the undecidability of solvability of distributed tasks by read/write memory \cite{Loop1}. Herlihy and Rajsbaum defined the \emph{algebraic signature} of a loop agreement task in terms of a group $G$ and an element $g \in G$. They proved that the algebraic signature completely characterizes the relative power of loop agreement tasks \cite{Loop2}, in the following sense. If tasks $T_1$ and $T_2$ have signatures $(G_1, g_1)$ and $(G_2, g_2)$, respectively, then $T_1$ implements $T_2$ exactly when there is a group homomorphism $\phi : G_1 \rightarrow G_2$ mapping $g_1$ to $g_2$. Thus the operational problem of loop agreement tasks implementing one another is reduced to an algebraic characterization.

In this paper, we describe how several loop agreement tasks can implement others, define compositions of loop agreement tasks, and extend the aforementioned algebraic characterization to these compositions. Roughly speaking, the \emph{composition} of $n$ loop agreement tasks is a task in which each process solves each of the $n$ tasks in parallel. We show that tasks $\{T_i\}$ with signatures $\{(G_i, g_i)\}$ solve $T$ with signature $(G, g)$ if and only if there is a homomorphism $\phi : G_1 \times \cdots \times G_n \rightarrow G$ mapping $(g_1, \ldots, g_n)$ to $g$. We also provide a means of replacing the loop agreement tasks $\{T_i\}$ with an equivalent task $\prod T_i$, called the \emph{composition} of the $\{T_i\}$. This composition of tasks is also a loop agreement task, and has relative power equivalent to that of all the $\{T_i\}$. That is, the $\{T_i\}$ implement $\prod T_i$ and $\prod T_i$ implements each $T_i$.

Finally, we take a category-theoretic approach to loop agreement in order to show that we have the correct notion of task composition. We define a category of loop agreement tasks, \textbf{Loop}, and show that the map assigning tasks to algebraic signatures is a functor into the category of pointed groups, \textbf{pGrp}. We also show that composition of loop agreement tasks is the categorical product in \textbf{Loop}, which strongly suggests that composition of tasks as defined in this paper correctly captures the operational meaning of parallel composition. We believe this category-theoretic approach may inspire future work on parallel composition of more general tasks beyond loop agreement, and may also inspire other work on applying category theory to general tasks.

Section 2 describes related work. Section 3 is a whirlwind tour of distributed tasks, algebraic topology, and loop agreement. Section 4 defines multiple implementation and composition of tasks and proves the main theorem. Section 5 provides an informal introduction to category theory and describes the category-theoretic view of loop agreement. Section 6 presents simple applications of our results, and in Section 7 we conclude with ideas for possible future work.

\section{Related Work}

Herlihy and Shavit introduced the use of algebraic topology \cite{Hmlgy1, Hmlgy2}, and in particular, homology theory to prove various impossibility results pertaining to set agreement and renaming. Since then homology theory has been used to prove other impossibility results in distributed computing \cite{OtherHm1, OtherHm2}. Gafni and Koutsoupias were the first to use the fundamental group in understanding distributed tasks \cite{Undecide} by showing the undecidability of wait-free solvability of certain tasks. Herlihy and Rajsbaum obtained similar undecidability results in other models which include loop agreement \cite{Loop1}, and also characterized the relative power of loop agreement tasks via their algebraic signatures \cite{Loop2}.

Loop agreement has also been generalized to higher dimensions. Liu, Xu, and Pan define \emph{$n$-rendezvous tasks} \cite{Rendez}, where processes begin on distinguished vertices of an embedded $(n-1)$-sphere of an $n$-dimensional complex, and converge on a simplex of the embedded sphere. They generalize the algebraic signature characterization to a subclass of rendezvous tasks called \emph{nice} rendezvous tasks, which are tasks whose output complexes have trivial homology groups below and above dimension $n$, and a free Abelian $n$-th homology group. The authors apply their main result to show there are countably infinite inequivalent nice rendezvous tasks.

Liu, Pu, and Pan explore a lower-dimensional variant of loop agreement called \emph{degenerate loop agreement} \cite{Degen}, which unlike loop agreement includes binary consensus. Processes begin on a $1$-dimensional complex, or a graph, and must converge to one of two possible starting locations in the graph. The authors prove that there are only two inequivalent tasks degenerate tasks: the trivial task and binary consensus. 

\section{Background}

In the first subsection, we describe the mathematical model used for distributed tasks, of which more details can be found in Herlihy, Kozlov, and Rajsbaum \cite{DCCT}. In the second subsection, we summarize important definitions and results from algebraic topology. 

\subsection{Distributed Computing}

Formally, a (colorless) \emph{task} is a triple $(\I, \Oh, \Gamma)$, where objects $\I$ and $\Oh$, called the \emph{input} and \emph{output complexes} of the task, are mathematical structures known as simplicial complexes. A \emph{simplicial complex} on a set $V$ is a collection of subsets $\C$ of $V$ such that $\C$ is downward closed under the subset relation. Complexes can be thought of as higher-dimensional graphs where ``edges" may ``connect" more than two vertices. In the context of tasks, vertices of $\I$ represent process input values, while simplexes of $\I$ represent valid input combinations. Likewise, vertices of $\Oh$ represent process output (or decision) values, and simplexes represent valid output combinations. Relating $\I$ and $\Oh$ is the map $\Gamma : \I \rightarrow 2^\Oh$, which is called the task's \emph{specification map}, and carries simplexes of $\I$ to subcomplexes of $\Oh$ in a monotonic way\footnote{In general, if $\A$ and $\B$ are simplicial complexes, then a function $\Phi : \A \rightarrow 2^\B$ is called a \emph{carrier map} if for each $\sigma \subseteq \tau \in \A$, $\Phi(\sigma)$ is a simplicial complex, and $\Phi(\sigma) \subseteq \Phi(\tau)$ (or $\Phi$ is \emph{monotonic}).}. The map $\Gamma$ associates each input combination with a set of legal output combinations.

Protocols are objects that solve tasks, and are also modeled by triples $(\I, \Prt, \Xi)$. As with tasks, $\I$ is the protocol's \emph{input complex}. The object $\Prt$ is also a simplicial complex, which is called the \emph{protocol complex}, and is similar to a task's output complex, but has a slightly different meaning. Rather than a final decision value, a vertex in $\Prt$ represents a process's uninterpreted state (or view) after running the protocol. The map $\Xi : \I \rightarrow 2^\Prt$, called the \emph{execution map}, is monotonic, and represents the possible states in which processes may result after running the protocol.

A \emph{simplicial map} $\dl : \I \rightarrow \Oh$ between two complexes is a vertex map that send simplexes to simplexes; that is, $\dl(\sigma) \in \Oh$ for each $\sigma \in \I$. A protocol $(\I, \Prt, \Xi)$ \emph{solves} $(\I, \Oh, \Gamma)$ if there exists a simplicial map $\dl : \Oh \rightarrow \Prt$, called a \emph{decision map}, that respects the task specification $\Gamma$. Formally, $\dl$ respects $\Gamma$ if for each simplex $\sigma \in \I$, we have $(\dl \circ \Xi)(\sigma) \subseteq \Gamma(\sigma)$.

Some tasks are inherently harder than others, and sometimes we can transform a protocol for one task into a protocol for another. We say task $T_1$ \emph{implements} $T_2$ if we can use the output complex of $T_1$ (or a subdivision of it) as a protocol complex for solving $T_2$. Mathematically speaking, if $T_1 = (\I, \Oh_1, \Gamma_1)$ and $T_2 = (\I, \Oh_2, \Gamma_2)$, then $T_1$ implements $T_2$ if there exists an natural number $N$ and a simplicial map $\phi : \Bary^N(\Oh_1) \rightarrow \Oh_2$ such that $(\phi \circ \Bary^N \circ \Gamma_1)(\sigma) \subseteq \Gamma_2(\sigma)$ for each $\sigma \in \I$. The barycentric subdivision operator $\Bary$ is a topological operator (see the next section) that models read/write memory. Two tasks are \emph{equivalent} if they implement each other.

\subsection{Algebraic Topology}

Before we can define loop agreement, we must briefly introduce the relevant machinery from algebraic topology. We assume a basic understanding of point-set topology. The algebraic topology used is at the undergraduate level, of which a formal treatment can be found in Hatcher \cite{Hatcher}. We begin with the formal definition of a simplicial complex.

\subsubsection{Simplicial Complexes}

\begin{defn}
Let $V$ be any set, whose elements are called \emph{vertices}. A \emph{simplicial complex} (over $V$) is a set of subsets $\C$ of $V$ such that for each set $\tau \in \C$, if $\sigma \subseteq \tau$, then $\sigma \in \C$. That is, $\C$ is downward closed under taking subsets. Elements of $\C$ are called \emph{simplexes}.
\end{defn}

We can think of simplicial complexes as a generalization of graphs, where simplexes may be incident to more than two vertices. Graphs are then precisely the simplicial complexes whose simplexes contain at most two vertices. Nontrivial graphs have dimension $1$, and in general, the \emph{dimension} of a complex $\C$ is $n - 1$, where $n$ is the size of the largest simplex in $\C$. The dimension of a simplex $\sigma$ is simply $|\sigma| - 1$. The \emph{standard $n$-simplex}, $\Delta^n$, is the simplicial complex on $n + 1$ vertices containing all possible simplexes. By convention, we will use $\{0, \ldots, n\}$ for the vertex set of $\Delta^n$.

A \emph{subcomplex} of $\C$ is a subset $\B \subseteq \C$ that is also a simplicial complex. For each nonnegative integer $k$, the \emph{$k$-skeleton} of $\C$, denoted $\skel^k(\C)$, is the subcomplex of $\C$ containing all simplexes of dimension at most $k$.

The above formulation of simplicial complexes defines them in a purely combinatorial way, but complexes can also be realized as topological spaces. Notationally, if $\C$ is a complex, then its geometric realization is denoted by $|\C|$. As previously mentioned, the \emph{barycentric subdivision} is an operator that models read/write memory, and is better understood geometrically than combinatorially. Given a geometric simplicial complex $|\C|$, we can create another geometric simplicial complex by adding new vertices to the barycenter of each simplex, and adding new simplexes accordingly. This gives rise to an abstract simplicial complex, denoted $\Bary(\C)$. Notice that the barycentric subdivision does not change the geometric realization of the original complex; that is, $|\Bary(\C)| = |\C|$.

The barycentric subdivision is also an important tool in approximating continuous functions with simplicial maps. If $f : |\A| \rightarrow |\B|$ is a continuous function between complexes, then a simplicial map $\phi : \A \rightarrow \B$ is called a \emph{simplicial approximation} of $f$ if for every $p \in |\A|$, $|\phi|(p)$ is contained in the smallest simplex containing $f(p)$. Using the barycentric subdivision, we can construct a simplicial approximation of any continuous function, as stated below.

\begin{fact}[Simplicial Approximation]

Let $f : |\A| \rightarrow |\B|$ be a continuous function between simplicial complexes. Then there exists an $N \in \N$ and a simplicial map $\phi : \Bary^N(\A) \rightarrow \B$ that is a simplicial approximation of $f$.

\end{fact}

We can take products of simplicial complexes. The product of two complexes is another complex that combines the structures of the original two.

\begin{defn}
\label{catprod}

Let $\C_1$ and $\C_2$ be simplicial complexes, and let $V(\C_1)$ and $V(\C_2)$ be their vertex sets, respectively. Then the \emph{(categorical) product of simplicial complexes} is a complex $\C_1 \times \C_2$ with vertex set $V(\C_1) \times V(\C_2)$. A subset $\sigma$ of $V(\C_1) \times V(\C_2)$ is a simplex in $\C_1 \times \C_2$ if and only if $\rho_1(\sigma)$ and $\rho_1(\sigma)$ are simplexes in $\C_1$ and $\C_2$, where $\rho_1$ and $\rho_2$ are projections onto the first and second coordinates, respectively.

\end{defn}

Intuitively, the product of complexes is a way of combining two complexes in the ``best possible way," and operationally, the product captures all possible combinations of process views if two tasks are solved in parallel. It is an important technical point that the product of complexes and product of topological spaces are not the same; it is not true that $|\A| \times |\B|$ and $|\A \times \B|$ are homeomorphic. They are, however, ``homotopy equivalent," which is a type of equivalence described in the next section.

To each topological space we can assign an invariant called the fundamental group, a basic construct taken from algebraic topology. The fundamental group is used to define the algebraic signature of a loop agreement task.

\subsubsection{Homotopy and the Fundamental Group}

Given a topological space $X$ and a basepoint $x_0 \in X$, a \emph{loop} in $X$ based at $x_0$ is a continuous function $\lambda : [0, 1] \rightarrow X$ such that $\lambda(0) = \lambda(1) = x_0$. Two loops $\lambda_1$ and $\lambda_2$ based $x_0$ are (loop) \emph{homotopic} if one loop can be continuously deformed to the other. More precisely, $\lambda_1$ and $\lambda_2$ are homotopic if there is a continuous function $H : [0, 1] \times [0, 1] \rightarrow X$ such that $H(0, -) = \lambda_1$, $H(1, -) = \lambda_2$, and $H(-, 0) = H(-, 1) = x_0$. Homotopy is an equivalence relation. We write $[\lambda]$ to denote the equivalence class of all loops homotopic to $\lambda$.

Let $\alpha : [0, 1] \rightarrow X$ and $\beta : [0, 1] \rightarrow X$ be two loops based at $x_0$. Then we can \emph{concatenate} $\alpha$ and $\beta$ to get another loop, $\alpha \cdot \beta$, defined by traversing $\alpha$, returning to $x_0$, and then traversing $\beta$. The loop $\alpha \cdot \beta : [0, 1] \rightarrow X$, also based at $x_0$, is defined as

\[(\alpha \cdot \beta)(t) = \left\{
     \begin{array}{ll}
        \alpha(2t) & \text{for } 0 \le t \le \frac{1}{2} \\
        \beta(2t - 1) & \text{for } \frac{1}{2} \le t \le 1 \\
     \end{array}
   \right. \]

Concatenation behaves well with homotopy. If $\alpha$ and $\beta$ are homotopic to $\alpha'$ and $\beta'$, respectively, then $[\alpha \cdot \beta] = [\alpha' \cdot \beta']$. From this it follows that concatenation is associative on classes of loops based at $x_0$. In fact, concatenation is a group operation on classes of loops based at $x_0$, with the inverse computed by traversing a loop in the opposite direction, and the identity element being the class of all loops homotopic to the constant loop at $x_0$. Formally, the inverse of $[\alpha]$ is the class of the loop $\alpha^{-1}(t) = \alpha(1 - t)$, and the class $[e]$ of loop $e(t) = x_0$ serves as the identity.

\begin{defn}

Let $X$ be a topological space, and let $x_0 \in X$ be a basepoint. Then the \emph{fundamental group} of $X$ at $x_0$, denoted $\pi_1(X, x_0)$, is the set of all loop homotopy classes with concatenation as its group operation. If $X$ is path-connected, then $\pi_1(X, x_0)$ is independent of $x_0$, and we simply write $\pi_1(X)$.

If $f : (X, x_0) \rightarrow (Y, y_0)$ is a basepoint-preserving continuous function, then $\pi_1$ also induces a group homomorphism $f_* : \pi(X, x_0) \rightarrow (Y, y_0)$ called the \emph{induced homomorphism}, defined by $f_*([\lambda]) = [f \circ \lambda]$.

\end{defn}

Henceforth, we assume all topological spaces and simplicial complexes under consideration are path-connected. For brevity, if $\C$ is a complex, we write $\pi_1(\C)$ instead of $\pi_1(|\C|)$. An important property of the fundamental group is how it behaves with the product of topological spaces.

\begin{fact} Let $X$ and $Y$ be topological spaces. Then $\pi_1(X \times Y) \cong \pi_1(X) \times \pi_1(Y)$.

\end{fact}

Homotopy is defined for loops, but it is more generally defined for continuous functions where the domain may not be $[0, 1]$. Two continuous functions $f, g : X \rightarrow Y$ are \emph{homotopic} if there is a continuous $H : X \times [0, 1] \rightarrow Y$ such that $H(-, 0) = f$ and $H(-, 1) = g$. We write $f \simeq g$ if this is the case. If in addition $X \subseteq Y$ and $H$ fixes $X$, then $H$ is called a \emph{deformation retraction} and we say $Y$ \emph{deformation retracts} onto $X$. If $\delta$ is a simplicial approximation of a continuous function $h$, then it is known that $|\delta| \simeq h$.

Using homotopy, we can define a weak equivalence between topological spaces called homotopy equivalence.

\begin{defn}

Let $X$ and $Y$ be topological spaces. Then $X$ and $Y$ are \emph{homotopy equivalent}, or $X \simeq Y$, if there are continuous functions $f : X \rightarrow Y$ and $g : Y \rightarrow X$ such that $g \circ f \simeq \id_X$ and $f \circ g \simeq \id_Y$. The maps $f$ and $g$ are called \emph{homotopy equivalences} and are \emph{homotopy inverses} of one another.

\end{defn}

Homeomorphic spaces are clearly homotopy equivalent. Homotopy equivalent spaces have the same fundamental group.

\begin{fact}

Let $X$ and $Y$ be topological spaces. If $X \simeq Y$, then $\pi_1(X) \cong \pi_1(Y)$.

\end{fact}

The next few facts are specifically about simplicial complexes. Recall that given two simplicial complexes $\A$ and $\B$, $|\A| \times |\B|$ and $|\A \times \B|$ are not topologically equivalent, though they are homotopy equivalent. See Kozlov's book on combinatorial algebraic topology for a detailed proof of this result \cite{CAT}.

\begin{fact}
\label{prodequiv}

Let $\A$ and $\B$ be simplicial complexes. Then $|\A| \times |\B| \simeq |\A \times \B|$.

\end{fact}

It follows that $|\A| \times |\B|$ and $|\A \times \B|$ have the same fundamental group. This will allow us to pass between the categorical product of $\A$ and $\B$ and the topological product of $|\A|$ and $|\B|$. We will require one more fact relating the fundamental group and the $2$-skeleton.

\begin{fact}
\label{only2}

Let $\C$ be a complex. Then the inclusion $\iota: \skel^2(\C) \rightarrow \C$ induces an isomorphism on fundamental groups.

\end{fact}

This fact can be derived from the following, more general result, which can be found in Hatcher \cite{Hatcher}. We call a continuous function $g : |\A| \rightarrow |\B|$ \emph{cellular} if $g$ maps skeleta to skeleta, or more precisely, if $g(|\skel^n(\A)|) \subseteq |\skel^n(\B)|$ for every $n$. Then every continuous $f : |\A| \rightarrow |\B|$ is homotopic to such a map $g$, as seen below.

\begin{fact}[Cellular Approximation]
\label{cellular}

Let $f : |\A| \rightarrow |\B|$ be a continuous function between simplicial complexes $\A$ and $\B$. Then $f$ is homotopic to a cellular function $g : |\A| \rightarrow |\B|$. Furthermore, if $\C \subseteq \A$ is a subcomplex such that $f$ is already cellular on $|\C|$, then we may require the homotopy between $f$ and $g$ to fix $|\C|$. 

\end{fact}

Now suppose we have a homotopy on a subcomplex and we want to extend it to the entire simplicial complex. The next fact, also found in Hatcher \cite{Hatcher}, allows us to do this.

\begin{fact}[Homotopy Extension]
\label{HEP}

Let $\C \subseteq \A$ and $\B$ be simplicial complexes, and let $F : |\A| \rightarrow |\B|$ be a continuous function. Suppose we have a homotopy $H : |\C| \times [0, 1] \rightarrow |\B|$ such that $H(-, 0) = F|_{|\C|}$. Then there is a homotopy extending $H$ to all of $|\A|$, respecting $F$. That is, we can find homotopy $H' : |\A| \times [0, 1] \rightarrow |\B|$ such that $H' | _{|\C| \times [0, 1]} = H$ and $H'(-, 0) = F$.

\end{fact}

\subsection{Loop Agreement}

We need a few more definitions before introducing loop agreement tasks.

\begin{defn}

Let $\C$ be a simplicial complex. An \emph{edge path} in $\C$ is an alternating sequence of vertices and edges, $v_1, e_1, v_2, e_2, \ldots, v_{k-1}, e_{k-1}, v_k$, where $e_i = \{v_i, v_{i+1}\}$. An \emph{edge loop} is an edge path with $v_0 = v_k$.

\end{defn}

\begin{defn}
Let $\C$ be a simplicial complex. Then a \emph{triangle loop} in $\C$ is a six-tuple $\lambda = (v_0, v_1, v_2, p_{01}, p_{12}, p_{20})$ such that each $v_i$ is a vertex in $\C$ and $p_{ij}$ is an edge path between $v_i$ and $v_j$.
\end{defn}

Triangle loops are indeed loops in the topological sense, but they can also be viewed as subcomplexes with designated vertices and edge paths. We now have the necessary tools and background to discuss loop agreement tasks. As previously stated, loop agreement is a class of tasks that models convergence of processes on an edge loop of a given space. The precise definition of loop agreement is given below \cite{Loop2}.

\begin{defn}

A \emph{loop agreement task} is a task $(\I, \Oh, \Gamma)$ for which $\I$ is the standard $2$-simplex, $\Oh$ is a (path-connected) $2$-dimensional simplicial complex with triangle loop $\lambda = (v_0, v_1, v_2, p_{01}, p_{12}, p_{20})$, and $\Gamma$ is defined as:

\[ \Gamma(\sigma) = \left\{
     \begin{array}{ll}
       \{v_i\} & : \sigma = \{i\} \\
       p_{ij} & : \sigma = \{i, j\} \\
       \Oh & : \sigma = \{0, 1, 2\} \\
     \end{array}
   \right. \]
   
\end{defn}

Notationally, we write $\Loop(\Oh, \lambda)$. Input vertices are carried to the designated vertices of $\ld$, the input edges are carried to paths between designated vertices, and the input triangle is carried to the whole output complex. The \emph{algebraic signature} of $\Loop(\Oh, \lambda)$ is $(\pi_1(\Oh), \lambda)$, and is used in the main theorem by Herlihy and Rajsbaum \cite{Loop2}:

\begin{thm}[Herlihy and Rajsbaum]
\label{original}
Task $\Loop(\K_1, \ld_1)$ implements $\Loop(\K_2, \ld_2)$ if and only if there exists a group homomorphism $h : \pi_1(\K_1) \rightarrow \pi_1(\K_2)$ such that $h([\ld_1]) = [\ld_2]$.
\end{thm}

The main contribution of this paper is parallel composition of tasks and the characterization of their relative power. We allow multiple tasks to implement another, and we generalize the above theorem to multiple tasks. We also show that a loop agreement task being implemented by two others is equivalent to the first being implemented by the composition of the second two.

\section{Composite Loop Agreement}

\subsection{Implementation by Multiple Tasks}

Informally, to implement one task by several others, we run protocols for each implementing task and use the combined output as a protocol complex. Given two loop agreement tasks, we will take the product of the output complexes and take the $2$-skeleton of the results; this becomes the output complex of the composite task. To obtain a loop in the output complex, we take the ``diagonal" of the product of the two original loops. We describe the construction of this loop in more detail.

\begin{defn}
Let $\ld_1 = (v_0, v_1, v_2, p_{01}, p_{12}, p_{20})$ and $\ld_2 = (w_0, w_1, w_2, q_{01}, q_{12}, q_{20})$ be triangle loops in complexes $\A$ and $\B$, respectively. Then the \emph{diagonal product} of $\ld_1$ and $\ld_2$, denoted $\ld_1 \star \ld_2$, is the triangle loop $(u_0, u_1, u_2, r_{01}, r_{12}, r_{20})$ in $\A \times \B$, where $u_i = (v_i, w_i)$. The path $r_{ij}$ is defined by traversing $p_{ij}$ while $w_i$ is fixed, followed by traversing $q_{ij}$ while $v_j$ is fixed. Note that we will use $p_{ij} \star q_{ij}$ to denote the path defined by $r_{ij}$ as above, though strictly speaking, the $\star$ operator denotes two different operations in $\ld_1 \star \ld_2$ and $p_{ij} \star q_{ij}$.

\end{defn}

\begin{defn}
Let $T_1 = \Loop(\K_1, \lambda_1)$, $T_2 = \Loop(\K_2, \lambda_2)$, and $T = \Loop(\K, \lambda)$ be loop agreement tasks. Let $\Gamma_1$, $\Gamma_2$, and $\Gamma$ be their respective specification maps. We say $T_1$ and $T_2$ \emph{implement} $T$ if there is an $N \in \N$ and a simplicial map $\phi : \Bary^N(\skel^2(\K_1 \times \K_2)) \rightarrow \K$ such that $(\phi \circ \Bary^N)(\skel^2(\Gamma_1(\sigma) \times \Gamma_2(\sigma))) \subseteq \Gamma(\sigma)$.
\end{defn}

Operationally, the participating processes first execute protocols for $T_1$ and $T_2$, ending up on a simplex of $\K_1 \times \K_2$. More precisely, because there are at most three participants, they end up on a simplex of $\skel^2(\K_1 \times \K_2)$. They then exchange results via $N$ rounds of reading and writing to ``scratchpad" read-write memory, ending up on a simplex of $\Bary^N (\skel^2(\K_1 \times \K_2))$. Finally, each process calls a decision map $\phi$ to choose a vertex in $\K$.

\subsection{Relative Power}

In this section we use the following notation for a continuous function that maps one triangle loop to another. If $\K_1$ and $\K_2$ are complexes with triangle loops $\lambda_1 = (v_0, v_1, v_2, p_{01}, p_{12}, p_{20})$ and $\lambda_2 = (w_0, w_1, w_2, q_{01}, q_{12}, q_{20})$, respectively, then we write $f : (\K_1, \lambda_1) \rightarrow (\K_2, \lambda_2)$ to denote a continuous function $f : |\K_1| \rightarrow |\K_2|$ such that $f(v_i) = w_i$ and $f(|p_{ij}|) \subseteq |q_{ij}|$.

We now state the main theorem of the paper.

\begin{thm}
\label{main}

Let $T_1 = \Loop(\K_1, \lambda_1)$, $T_2 = \Loop(\K_2, \lambda_2)$, and $T = \Loop(\K, \lambda)$. Then $T_1$ and $T_2$ implement $T$ if and only if there exists a group homomorphism $h : \pi_1(\K_1) \times \pi_1(\K_2) \rightarrow \pi_1(\K)$ such that $h([\lambda_1], [\lambda_2]) = [\lambda]$.

\end{thm}

Theorem \ref{main} describes only two loop agreement tasks implementing a third, but by finite induction, one can easily generalize this to $n$ tasks. Its proof is broken down into two other theorems, which jointly prove Theorem \ref{main}. The first theorem is a topological characterization of two tasks implementing a third, while the second theorem is on the correspondence between continuous functions and group homomorphisms.

\begin{thm}
\label{imptop}
Tasks $T_1$ and $T_2$ implement $T$ if and only if there exists a continuous function $f : (\skel^2(\K_1 \times \K_2), \ld_1 \star \ld_2) \rightarrow (\K, \ld)$.
\end{thm}

We prove Theorem \ref{imptop} by proving each direction individually via the following lemmas.

\begin{lem}
\label{imptopfwd}
If there is a continuous function $f : (\skel^2(\K_1 \times \K_2), \ld_1 \star \ld_2) \rightarrow (\K, \ld)$, then $T_1$ and $T_2$ implement $T$.
\end{lem}

\begin{proof}
Suppose such a function $f$ exists, and let $\Gamma_1$, $\Gamma_2$, and $\Gamma$ be the specification maps for $T_1$, $T_2$, and $T$, respectively. To prove $T_1$ and $T_2$ implement $T$, we require an $N \in \N$ and a simplicial map $\phi : \Bary^N(\skel^2(\K_1 \times \K_2)) \rightarrow \K$ such that for each $\sigma \in \I$, we have $(\phi \circ \Bary^N)(\skel^2(\Gamma_1(\sigma) \times \Gamma_2(\sigma))) \subseteq \Gamma(\sigma)$. We will construct such a $\phi$ by taking a simplicial approximation of a suitably defined continuous function.

Let $p_{01}$, $p_{12}$, and $p_{20}$, and $q_{01}$, $q_{12}$, and $q_{20}$ be the designated edge paths of $\ld_1$ and $\ld_2$, respectively. Consider $X = |(p_{01} \times q_{01})| \cup |(p_{12} \times q_{12})| \cup |(p_{20} \times q_{20})| \subseteq |\K_1 \times \K_2|$ as a topological subspace. Clearly, each $|p_{ij} \times q_{ij}|$ deformation retracts to the corresponding path $|p_{ij} \star q_{ij}|$ in $|\lambda_1 \star \lambda_2|$. In other words, we have a continuous function $H : X \times [0, 1] \rightarrow |\K_1 \times \K_2|$ such that $H(x, 0) = x$, $H(X, 1) = |\ld_1 \star \ld_2|$, and $H(a, t) = a$ for each $a \in |\ld_1 \star \ld_2|$, $x \in X$, and $t \in [0, 1]$. Now using Fact \ref{HEP}, we can extend $H$ to a continuous function $H' : |\K_1 \times \K_2| \times [0, 1] \rightarrow |\K_1 \times \K_2|$. In particular, define $r : |\K_1 \times \K_2| \rightarrow |\K_1 \times \K_2|$ as $r(x) = H(x, 1)$. This is a continuous function from $|\K_1 \times \K_2|$ to itself that fixes $|\ld_1 \star \ld_2|$ while collapsing $X$ to $|\ld_1 \star \ld_2|$. We restrict $r$ to $|\skel^2(\K_1 \times \K_2)|$ and invoke Fact \ref{cellular} to get a function $g : |\skel^2(\K_1 \times \K_2)| \rightarrow |\skel^2(\K_1 \times \K_2)|$ that fixes $|\ld_1 \star \ld_2|$ while collapsing $\skel^2(X)$ to $|\ld_1 \star \ld_2|$. Now let $F = f \circ g$. This is a continuous function $F : |\skel^2(\K_1 \times \K_2)| \rightarrow |\K|$ which maps $\ld_1 \star \ld_2$ to $\ld$. 

To show $F$ is carried by $\Gamma$, first consider the case where $|\sigma| = 1$. Then the point $|\Gamma_1(\sigma) \times \Gamma_2(\sigma)|$ is contained in $|\lambda_1 \star \lambda_2|$, so is fixed under $g$, and hence mapped to the appropriate point in $\ld$ by the given function $f$. The case $|\sigma| = 2$ is similar. We have $|\Gamma_1(\sigma) \times \Gamma_2(\sigma)| \subseteq X$, which collapses to $|\lambda_1 \star \lambda_2|$ under $g$. The function $f$ maps this to $\ld$, as desired. The final case is when $|\sigma| = 3$, which does not require any part of the proof above, since $\Gamma(\sigma) = \K$. In all cases, we see that $F$ is carried by $\Gamma$. Letting $\phi : \Bary^N(\skel^2(\K_1 \times \K_2)) \rightarrow \K$ be a simplicial approximation of $F$, $\phi$ is also carried by $\Gamma$, so we have the required decision map.

\end{proof}

\begin{lem}
\label{imptopbk}
If tasks $T_1$ and $T_2$ implement $T$, then there is a continuous function $f : (\skel^2(\K_1 \times \K_2), \ld_1 \star \ld_2) \rightarrow (\K, \ld)$.
\end{lem}

\begin{proof}
Assuming $T_1$ and $T_2$ implement $T$, we have a simplicial map $\phi: \Bary^N(\skel^2(\K_1 \times \K_2)) \rightarrow \K$ that is carried by $\Gamma$. In particular, $\phi$ maps $\ld_1 \star \ld_2$ to $\ld$. Let $f : (\skel^2(\K_1 \times \K_2), \ld_1 \star \ld_2) \rightarrow (\K, \ld)$, defined by $f(x) = |\phi|(x)$. Then $f$ maps $|\ld_1 \star \ld_2|$ to $|\ld|$ since $\phi$ does this as well.
\end{proof}

Lemmas \ref{imptopfwd} and \ref{imptopbk} together prove Theorem \ref{imptop}. Next, we prove the correspondence between continuous functions and group homomorphisms. In order to do this, we refer to the following result shown in Herlihy and Rajsbaum \cite{Loop2}.

\begin{lem}
\label{2dim}

Let $\K$ and $\Ll$ be finite, connected, $2$-dimensional simplicial complexes, and let $h : \pi_1(\K) \rightarrow \pi_1(\Ll)$ be a homomorphism with $h([\sigma]) = [\tau]$. Then there exists a continuous $f : |\K| \rightarrow |\Ll|$ such that $f_* = h$ and $f \circ \sigma = \tau$.

\end{lem}

\begin{thm}
\label{topalg}

There exists a continuous function $f : (\skel^2(\K_1 \times \K_2), \ld_1 \star \ld_2) \rightarrow (\K, \ld)$ if and only if there exists a group homomorphism $h : \pi_1(\K_1) \times \pi_1(\K_2) \rightarrow \pi_1(\K)$ such that $h([\ld_1], [\ld_2]) = [\ld]$.

\end{thm}

\begin{proof}

First suppose we have a continuous function $f : (\skel^2(|\K_1 \times \K_2|), \ld_1 \star \ld_2) \rightarrow (\K, \ld)$. We begin by constructing a homomorphism $h' : \pi_1(|\K_1 \times \K_2|) \rightarrow \pi_1(\K)$ with $h'([\ld_1 \star \ld_2]) = [\ld]$. Let $\iota : \skel^2(|\K_1 \times \K_2|) \rightarrow |\K_1 \times \K_2|$ be the inclusion map, whose induced homomorphism is actually an isomorphism, by Fact \ref{only2}. Then we let $h' = f_* \circ \iota_*^{-1}$. In order to show $h'([\ld_1 \star \ld_2]) = [\ld]$, it suffices to show that $\iota_*^{-1}([\ld_1 \star \ld_2]) = [\ld_1 \star \ld_2]$. However, notice that $[\ld_1 \star \ld_2] = \iota_*([\ld_1 \star \ld_2])$ since $\ld_1 \star \ld_2$ is already in $\skel^2(|\K_1 \times \K_2|)$, so $\iota_*^{-1}([\ld_1 \star \ld_2]) = [\ld_1 \star \ld_2]$ as required.

Now, we define the desired homomorphism $h : \pi_1(\K_1) \times \pi_1(\K_2) \rightarrow \pi_1(\K)$ using $h'$.  Let $\alpha_1$ and $\alpha_2$ be loops in $\K_1$ and $\K_2$ respectively. By Fact \ref{cellular}, $\alpha_1$ and $\alpha_2$ are homotopic to edge loops $\beta_1$ and $\beta_2$. Now define $h$ as $h([\alpha_1], [\alpha_2]) = h'([\beta_1 \star \beta_2])$. Then it follows that $h([\lambda_1], [\lambda_2]) = [\lambda]$. To show $h'$ is well-defined, we need to show that $|\beta_1 \star \beta_2| \simeq |\beta_1' \star \beta_2'|$ for other edge-loop representatives $\beta_1'$ and $\beta_2'$ of $\alpha_1$ and $\alpha_2$. We can find edge homotopies $H_1$ and $H_2$ taking $\beta_1$ and $\beta_2$ to $\beta_1'$ and $\beta_2'$, respectively, so $H_1 \star H_2$ is an edge homotopy from $|\beta_1 \star \beta_2| \simeq |\beta_1' \star \beta_2'|$, proving that $h$ is well-defined. We have thus found the required $h$, which proves the forward direction of the theorem.

Now suppose we start with a homomorphism $h$ as described above. We reverse the above argument. We begin by constructing a homomorphism $h' : \pi_1(|\K_1 \times \K_2|) \rightarrow \pi_1(\K)$. Let $\alpha$ be a loop in $|\K_1 \times \K_2|$. As before, $\alpha$ is homotopic to some edge loop $\beta$ of $\K_1 \times \K_2$. We define $h'([\alpha]) = h([\rho_1 \circ \beta], [\rho_2 \circ \beta]])$, where the $\rho_i$ are the projection maps. This map is clearly well-defined and a homomorphism since it is the composition of $h$ and the induced maps of the $\rho_i$.

Now we define a homomorphism $h'' : \pi_1(\skel^2(|\K_1 \times \K_2|)) \rightarrow \pi_1(\K)$ with $h''([\ld_1 \star \ld_2]) = [\ld]$, using $h'$. Let $\iota$ be the inclusion map, as before. Then we define $h'' = h' \circ \iota_*$. Since $\iota_*([\ld_1 \star \ld_2]) = [\ld_1 \star \ld_2]$, we see that $h''([\ld_1 \star \ld_2]) = [\ld]$. Finally, we invoke Lemma \ref{2dim} on $h''$ to obtain the required $f$. This proves the backward direction of the theorem, and completes the proof.

\end{proof}

Theorems \ref{imptop} and \ref{topalg} together prove Theorem \ref{main}.

\subsection{Composite Loop Agreement}

In defining multiple implementation, we said that tasks $T_1$ and $T_2$ implement $T$ if we can use the combined output complex $\skel^2(\K_1 \times \K_2)$ of $T_1$ and $T_2$ to solve $T$. We can think of parallel execution of protocols for $T_1$ and $T_2$ as solving a task with input complex $\Delta^2$, output complex $\skel^2(\K_1 \times \K_2)$, and specification $\Gamma_1 \times \Gamma_2$. We get a task $T' = (\Delta^2, \skel^2(\K_1 \times \K_2), \Gamma_1 \times \Gamma_2)$, and from the definitions it is clear that $T_1$ and $T_2$ implement $T$ if and only if $T'$ implements $T$. Unfortunately, $T'$ is not a loop agreement task, since processes starting on an edge in $\Delta^2$ can land on any edge in $\lambda_1 \times \lambda_2$ and still obey the task specification. However, the subcomplex $\lambda_1 \times \lambda_2$ is not a loop. We address this by defining a loop agreement task $T_1 \times T_2$ with output complex $\skel^2(\K_1 \times \K_2)$ and triangle loop $\lambda_1 \star \lambda_2$. We then show that $T'$ and $T_1 \times T_2$ implement one another, so are equivalent.

\begin{defn}

Let $T_1 = \Loop(\K_1, \lambda_1)$ and $T_2 = \Loop(\K_2, \lambda_2)$ be loop agreement tasks. Then the \emph{composition} of $T_1$ and $T_2$, denoted $T_1 \times T_2$, is the loop agreement task $\Loop(\skel^2(\K_1 \times \K_2), \lambda_1 \star \lambda_2)$.

\end{defn}

\begin{prop}

Tasks $T_1$ and $T_2$ implement $T_1 \times T_2$.

\end{prop}

\begin{proof}

This is an immediate consequence of Lemma \ref{imptopfwd}.

\end{proof}

\begin{prop}
\label{taskproj}

Task $T_1 \times T_2$ implements $T_1$ (respectively $T_2$).

\end{prop}

\begin{proof}

Lemma 6.2 from Herlihy and Rajsbuam \cite{Loop2} states that it suffices to show there is a continuous function $f : \skel^2(\K_1 \times \K_2) \rightarrow \K_1$ mapping $\lambda_1 \star \lambda_2$ to $\lambda_1$. It is easy to see that the projection map $\rho_1 : \skel^2(\K_1 \times \K_2) \rightarrow \K_1$ satisfies this condition. The proof that $T_1 \times T_2$ implements $T_2$ is identical.

\end{proof}

\section{Category Theory of Loop Agreement}

In this section, we describe a more formal connection between the class of loop agreement tasks and the class of groups, using the language of category theory. We formalize the correspondence between loop agreement tasks and algebraic signatures, and also state one direction of the main theorem using category-theoretic formalism. Intuitively, loop agreement tasks form an organized collection of objects called a ``category", with decision maps, or ``morphisms", connecting two tasks if one implements the other. The algebraic signature assignment, an example of a ``functor" between categories, transforms the loop agreement category into a category of groups. The composition of loop agreement tasks as defined in this paper is actually their ``categorical" product.

We begin with some necessary background in category theory; see Mac Lane \cite{MacLane} for a rigorous treatment.

\subsection{Categories}

A \emph{category} $C$ consists of a collection of \emph{objects}, denoted $\Ob(C)$, and a collection of \emph{morphisms} between those objects, denoted $\Hom(C)$. Each morphism has a \emph{domain} and \emph{codomain}, which are both objects in $\Ob(C)$. If $f$ is a morphism with domain $X$ and codomain $Y$, then we write $f : X \rightarrow Y$. This notation is suggestive of set functions, and indeed the category of sets is a well-known category, and has sets as objects and set functions as morphisms.

As with ordinary functions, morphisms can be composed. Formally, $\Hom(C)$ is equipped with a binary operation called \emph{composition}. If $f$ and $g$ are morphisms, then their composition is denoted $f \circ g$. Note that composition of functions is only defined when the codomain of the first morphism is equal to the domain of the second. Composition is required to be associative; that is, given $f : W \rightarrow X$, $g : X \rightarrow Y$, and $h : Y \rightarrow Z$, we must have $h \circ (g \circ f) = (h \circ g) \circ f$. Composition also requires an identity morphism for each object $X$, denoted $\id_X$, such that for each $f : X \rightarrow Y$, we have $f \circ \id_X = f = \id_Y \circ f$.

As mentioned, sets and set functions comprise the category of sets, denoted \textbf{Set}. Another example is the category of topological spaces, where objects are spaces and morphisms are continuous functions between them, and is denoted \textbf{Top}. There is also the category of groups, \textbf{Grp}, consisting of groups and groups homomorphisms. Algebraic signatures belong to a similar category called the category of \emph{pointed} groups, \textbf{pGrp}, whose objects are groups with distinguished elements and whose morphisms are group homomorphisms that preserve distinguished elements.

We say one category $C$ is a \emph{subcategory} of another category $C'$ if the objects and morphisms of $C$ are contained in $C'$. For example, the category of Abelian groups, $\textbf{Ab}$, is a subcategory of $\textbf{Grp}$. We will also make use of $\textbf{SimC}_n$, which is the subcategory of $\textbf{SimC}$ containing all simplicial complexes of dimension up to $n$ and the morphisms between them.

We can transform objects and morphisms of one category to objects and morphisms of another. Given categories $C$ and $D$, a \emph{functor} $F : C \rightarrow D$ assigns to each object $X \in \Ob(C)$ an object $F(X) \in \Ob(D)$, and to each morphism $f : X \rightarrow Y$ a morphism $F(f) : F(X) \rightarrow F(Y)$. Functors must respects composition; that is, given two compatible morphisms $f, g \in \Hom(C)$, we must have $F(f \circ g) = F(f) \circ F(g)$. Functors must also respect identity morphisms: $F(\id_X) = \id_{F(X)}$. A common example of a functor is the fundamental group functor $\pi_1 : \textbf{pTop} \rightarrow \textbf{Grp}$, which maps pointed topological spaces to their respective fundamental groups, and maps continuous functions to their induced homomorphisms. If we consider only path-connected spaces, then $\pi_1$ is also a functor from \textbf{Top} to \textbf{Grp}. The geometric realization $|\cdot| : \textbf{SimC} \rightarrow \textbf{Top}$ is a functor from the category of simplicial complexes with simplicial maps to \textbf{Top}, which maps complexes and simplicial maps to their respective geometric realizations.

We can also combine two objects from a category to produce a new one, which is an operation called the categorical product. The categorical product of two objects is the most general object that maps onto the original two.

\begin{defn}

Let $C$ be a category, and let $X_1$ and $X_2$ be objects in this category. The \emph{categorical product} of $X_1$ and $X_2$ is the unique object $X_1 \times X_2$ satisfying the following: there exist morphisms (called \emph{projections}) $\rho_1 : X_1 \times X_2 \rightarrow X_1$ and $\rho_2 : X_1 \times X_2 \rightarrow X_2$ such that for any object $X$ with morphisms $f_1 : X \rightarrow X_1$ and $f_2 : X \rightarrow X_2$, there exists a unique morphism $f : X \rightarrow X_1 \times X_2$ such that $f_1 = \rho_1 \circ f$ and $f_2 = \rho_2 \circ f$. That is, $f_1$ and $f_2$ factor through $X_1 \times X_2$ in a unique way, via $f$. The morphism $f$ is called the \emph{product morphism} of $f_1$ and $f_2$.

\end{defn}

Examples of categorical products include the product topology for topological spaces \cite{MacLane}, the direct product of groups, and the categorical product of simplicial complexes as stated in Definition \ref{catprod} \cite{CAT}.

\subsection{The Category of Loop Agreement Tasks}

Now that we have the preliminaries of category theory, we define \textbf{Loop}, the category of loop agreement tasks. We let $\Ob(\textbf{Loop})$ be the collection of all loop agreement tasks $\Loop(\K, \lambda)$, where $\K$ ranges over all finite connected $2$-dimensional complexes and $\lambda$ ranges over all edge loops. Morphisms in \textbf{Loop} are valid decision maps between tasks. That is, given tasks $T_1 = \Loop(\K_1, \lambda_1)$ and $T_2 = \Loop(\K_2, \lambda_2)$, a morphism $f : T_1 \rightarrow T_2$ is a pair $(\delta, N)$ where $N \in \N$ and $\delta : \Bary^N(\K_1) \rightarrow \K_2$ is a decision map such that $T_1$ solves $T_2$ via $\delta$. Composition of morphisms is defined as follows. Given objects $T_1 = \Loop(\K_1, \lambda_1)$, $T_2 = \Loop(\K_2, \lambda_2)$, $T_3 = \Loop(\K_3, \lambda_3)$, and morphisms $f_1 : T_1 \rightarrow T_2$, $f_2 : T_2 \rightarrow T_3$ where $f_1 = (\delta_1, N_1)$ and $f_2 = (\delta_2, N_2)$, the composition $f_2 \circ f_1$ is defined as $(\delta_2 \circ \Bary^{N_2}(\delta_1), N_1 + N_2)$. Two morphisms are considered equivalent if their simplicial maps are homotopic\footnote{By identifying morphisms (in this case homotopic ones), we are constructing a \emph{quotient category} from the original one. In order to construct a quotient category, the equivalence must be compatible with composition. However, it is well known that homotopy is compatible with compositions of continuous functions.}. We must now prove that \textbf{Loop} is a category.

\begin{thm}

\emph{\textbf{Loop}} is a category.

\end{thm}

\begin{proof}

Let $T_i$ and $f_i$ be defined as above, and let $\Gamma_i$ be the tasks' respective specification maps. To show \textbf{Loop} is a category, we need to show that $\Hom(\textbf{Loop})$ is closed under composition, composition is associative, and identity morphisms exist. Showing that $\Hom(\textbf{Loop})$ is closed under composition amounts to showing that $T_1$ solves $T_3$ via $\delta_2 \circ \Bary^{N_2}(\delta_1) : \Bary^{N_1 + N_2}(\K_1) \rightarrow \K_3$. For brevity we define $\delta = \delta_2 \circ \Bary^{N_2}(\delta_1)$.

From the definition of task implementation, we know that $\delta_1 \circ \Bary^{N_1} \circ \Gamma_1 \subseteq \Gamma_2$ and $\delta_2 \circ \Bary^{N_2} \circ \Gamma_2 \subseteq \Gamma_3$, and we want to show $\delta \circ \Bary^{N_1 + N_2} \circ \Gamma_1 \subseteq \Gamma_3$. So $\delta_2 \circ \Bary^{N_2} \circ \delta_1 \circ \Bary^{N_1} \circ \Gamma_1 \subseteq \delta_2 \circ \Bary^{N_2} \circ \Gamma_2 \subseteq \Gamma_3$. We know that $\Bary^{N_2} \circ \delta_1 = \Bary^{N_2} (\delta_1) \circ \Bary^{N_2}$, so $\delta_2 \circ \Bary^{N_2} \circ \delta_1 \circ \Bary^{N_1} \circ \Gamma_1 = \delta_2 \circ \Bary^{N_2}(\delta_1) \circ \Bary^{N_2} \circ \Bary^{N_1} \circ \Gamma_1 = \delta \circ \Bary^{N_1 + N_2} \circ \Gamma_1 \subseteq \Gamma_3$. Therefore $T_1$ solves $T_3$ via $\delta$, so $\Hom(\textbf{Loop})$ is closed under our definition of composition.

Verifying associativity follows a similar argument. Again, let $T_i$ and $f_i$ be defined as above, and in addition let $T_4 = \Loop(\K_4, \ld_4)$ and let $f_3 : T_3 \rightarrow T_4$ with $f_3 = (\delta_3, N_3)$. We must show that $(f_3 \circ f_2) \circ f_1 = f_3 \circ (f_2 \circ f_1)$. But $(f_3 \circ f_2) \circ f_1 = (\delta_3 \circ \Bary^{N_3}(\delta_2), N_2 + N_3) \circ (\delta_1, N_1) = (\delta_3 \circ \Bary^{N_3}(\delta_2) \circ \Bary^{N_2 + N_3}(\delta_1), N_1 + N_2 + N_3)$, and $f_3 \circ (f_2 \circ f_1) = (\delta_3, N_3) \circ (\delta_2 \circ \Bary^{N_2}(\delta_1), N_1 + N_2) = (\delta_3 \circ \Bary^{N_3}(\delta_2 \circ \Bary^{N_2}(\delta_1)), N_1 + N_2 + N_3) = (\delta_3 \circ \Bary^{N_3}(\delta_2) \circ \Bary^{N_2 + N_3}(\delta_1), N_1 + N_2 + N_3)$, so $(f_3 \circ f_2) \circ f_1 = f_3 \circ (f_2 \circ f_1)$. Therefore composition is associative.

The last requirement, existence of identity morphisms, is trivial to show. Task $T_1$ solves itself via the decision map $(\id_{\K_1}, 0)$. This finishes the proof that $\textbf{Loop}$ is a category.

\end{proof}

Next, we show that the algebraic signature of Herlihy and Rajsbaum can be formulated as a functor between $\textbf{Loop}$ and $\textbf{pGrp}$.

\begin{defn}

Let $T_1, T_2 \in \Ob(\textbf{Loop})$ with $T_1 = \Loop(\K_1, \ld_1)$ and $T_2 = \Loop(\K_2, \ld_2)$, and let $f_1 : T_1 \rightarrow T_2$ with $f_1 = (\delta_1, N_1)$ be a morphism between the two. Then the \emph{algebraic signature functor} is a functor $S : \textbf{Loop} \rightarrow \textbf{pGrp}$ defined as follows. Object $T_1$ is mapped to $(\pi_1(\K_1), [\ld_1])$, while morphism $f_1 : T_1 \rightarrow T_2$ is mapped to $|\delta_1|_* : (\pi_1(\K_1), [\ld_1]) \rightarrow (\pi_2(\K_2), [\ld_2])$.

\end{defn}

\begin{thm}

$S : \emph{\textbf{Loop}} \rightarrow \emph{\textbf{pGrp}}$ is a functor.

\end{thm}

\begin{proof}

We use the fact that $\pi_1$ and $|\cdot|$ are both functors. We need to show that $S$ preserves identity morphisms and respects composition of morphisms. Let $T_1$, $T_2$, and $f$ be defined as above, and let $T_3 = \Loop(\K_3, \ld_3)$ and let $f_2 : T_2 \rightarrow T_3$ with $f_2 = (\delta_2, N_2)$. Then, using the functoriality of $\pi_1$ and $|\cdot|$, we have $S(f_2 \circ f_1) = S((\delta_2 \circ \Bary^{N_1}(\delta_1), N_1 + N_2)) = |\delta_2 \circ \Bary^{N_1}(\delta_1)|_* = (|\delta_2| \circ |\Bary^{N_1}(\delta_1)|)_* = |\delta_2|_* \circ |\delta_1|_* = S(f_2) \circ S(f_1)$, so $S$ respects composition. Now let $\id_{T_1}$ be the identity morphism of $T_1$. Then $S(\id_{T_1}) = S((\id_{\K_1}, 0)) = |\id_{\K_1}|_* = \id_{\pi_1(\K_1)}$, so $S$ also preserves identity morphisms. $S$ is well-defined since $\pi_1$ cannot distinguish between homotopic functions. We conclude that $S$ is a functor.

\end{proof}

We are almost ready to prove that composition of loop agreement tasks is in fact the categorical product in $\textbf{Loop}$, but first we need a lemma describing the categorical product in $\textbf{SimC}_2$, which is slightly different than the one in $\textbf{SimC}$.

\begin{lem}
\label{catprod2} 

If $\K_1$ and $\K_2$ are objects in $\emph{\textbf{SimC}}_2$, then $\skel^2(\K_1 \times \K_2)$ is their categorical product in $\emph{\textbf{SimC}}_2$.

\end{lem}

\begin{proof}

We first define projection maps $\rho_1 : \skel^2(\K_1 \times \K_2) \rightarrow \K_1$ and $\rho_2 : \skel^2(\K_1 \times \K_2) \rightarrow \K_2$ as $\rho_1(v_1, v_2) = v_1$ and $\rho_2(v_1, v_2) = v_2$. That is, the $\rho_i$ are the restrictions to the $2$-skeleton of the projection maps found in Definition \ref{catprod}, so they are clearly simplicial.

Now suppose we have a $2$-dimensional complex $\K$ with simplicial maps $\delta_1 : \K \rightarrow \K_1$ and $\delta_2 : \K \rightarrow \K_2$. Then we define $\delta : \K \rightarrow \skel^2(\K_1 \times \K_2)$ as $\delta(v) = (\delta_1(v), \delta_2(v))$. This is the only possible set function $\delta$ that makes the diagram commute; that is, $\delta$ is the only set function such that $\delta_1 = \rho_1 \circ \delta$ and $\delta_2 = \rho_2 \circ \delta$. This proves uniqueness, but we must also show that $\delta$ is simplicial.

Let $\sigma$ be a simplex in $\skel^2(\K_1 \times \K_2)$. Then $\delta_1(\sigma)$ and $\delta_2(\sigma)$ are simplexes in $\K_1$ and $\K_2$, respectively. But as we have shown, $\delta_1(\sigma) = \rho_1(\delta(\sigma))$ and $\delta_2(\sigma) = \rho_2(\delta(\sigma))$, so in particular, we see that $\rho_1(\delta(\sigma))$ and $\rho_2(\delta(\sigma))$ are simplexes. Hence by Definition \ref{catprod}, $\delta(\sigma)$ is a simplex in $\K_1 \times \K_2$, and furthermore it is a simplex in $\skel^2(\K_1 \times \K_2)$ since the dimension of $\sigma$ is at most $2$. So $\delta$ is a simplicial map, which proves that $\skel^2(\K_1 \times \K_2)$ is the categorical product of $\K_1$ and $\K_2$ in $\textbf{SimC}_2$.

\end{proof}

Note that Lemma \ref{catprod2} easily generalizes to $\textbf{SimC}_n$ and the $n$-skeleton.

\begin{thm}

Composition of loop agreement tasks is the categorical product in \emph{\textbf{Loop}}.

\end{thm}

\begin{proof}

Let $T_1 = \Loop(\K_1, \lambda_1)$ and $T_2 = \Loop(\K_2, \lambda_2)$ be tasks as defined before, and let $\Gamma_1$ and $\Gamma_2$ be their specification maps, respectively. Let $\Gamma_\times$ be the specification map of $T_1 \times T_2$. We must first define decision maps from $T_1 \times T_2$ to $T_1$ and $T_2$ that would make $T_1 \times T_2$ the categorical product. We know that $\skel^2(\K_1 \times \K_2)$ is the categorical product of $\K_1$ and $\K_2$ in the category $\textbf{SimC}_2$, and that the product comes with projection maps $\rho_1 : \skel^2(\K_1 \times \K_2) \rightarrow \K_1$ and $\rho_2 : \skel^2(\K_1 \times \K_2) \rightarrow \K_2$. Using these, we define maps $g_1 : T_1 \times T_2 \rightarrow T_1$ and $g_2 : T_1 \times T_2 \rightarrow T_2$ with $g_1 = (\rho_1, 0)$ and $g_2 = (\rho_2, 0)$, and we claim that these maps make $T_1 \times T_2$ the categorical product of $T_1$ and $T_2$.

First, we must show that $g_1$ and $g_2$ are decision maps solving $T_1$ and $T_2$. However, we already showed this in Proposition \ref{taskproj}. To prove that $g_1$ and $g_2$ are the projection maps that make $T_1 \times T_2$ the categorical product, we consider a task $T$ that implements both $T_1$ and $T_2$, say via maps $f_1 = (\delta_1, N_1)$ and $f_2 = (\delta_2, N_2)$, respectively. Let $T = \Loop(\K, \lambda)$ and let $\Gamma$ be its specification map. We must find a decision map that solves $T_1 \times T_2$ from $T$. Without loss of generality, assume $N_1 \ge N_2$, so let $\delta_2' : \Bary^{N_1}(\K) \rightarrow \K_2$ be a simplicial approximation of $\delta_2$. Then $\delta = (\delta_1, \delta_2')$ is a map from $\Bary^{N_1}(\K)$ to $\skel^2(\K_1 \times \K_2)$, though it does not necessarily carry $\lambda$ to $\lambda_1 \star \lambda_2$. Instead, $g = (\delta, N_1)$ is a morphism from $\Loop(\K, \lambda)$ to $\Loop(\skel^2(K_1 \times \K_2), \delta(\lambda))$. However, it is easy to see that $\delta(\lambda)$ is homotopic to $\lambda_1 \star \lambda_2$. Using Fact \ref{HEP}, we can extend this to a homotopy on all of $\skel^2(\K_1 \times \K_2)$, so we obtain a continuous function $h : |\skel^2(\K_1 \times \K_2)| \rightarrow |\skel^2(\K_1 \times \K_2)|$. Let $\gamma : \Bary^M(\skel^2(\K_1 \times \K_2)) \rightarrow \skel^2(\K_1 \times \K_2)$ be a simplicial approximation of $h$. Then notice that $g' = (\gamma, M)$ is a morphism from $\Loop(\skel^2(\K_1 \times \K_2), \delta(\lambda))$ to $\Loop(\skel^2(\K_1 \times \K_2), \lambda_1 \star \lambda_2)$. So $f = g' \circ g$ is a morphism $f : T \rightarrow T_1 \times T_2$. We must also show that $f = (\gamma \circ \Bary^M(\delta), N_1 + M)$ makes the diagram commute. Let $\delta' = \gamma \circ \Bary^M(\delta)$. We know that $\rho_i \circ \delta \simeq \delta_i$ by construction of $\delta$, and it is also clear that $\delta' \simeq \delta$, by construction of $\delta'$ and $\gamma$. It follows that $\rho_i \circ \delta' \simeq \delta_i$, proving that $f$ makes the diagram commute. Thus we have the required product morphism.

Finally, it remains to show that $f$ is unique. Let $f'$ be any such morphism making the diagram commute, and let $\delta'$ be its simplicial map. Then, as set maps, we know that $\delta' = (\rho_1 \circ \delta', \rho_2 \circ \delta')$. However, we are assuming that $|\rho_1 \circ \delta'| \simeq |\delta_1|$ and $|\rho_2 \circ \delta'| \simeq |\delta_2|$, so this allows us to conclude that $|\delta'| = (|\rho_1 \circ \delta'|, |\rho_2 \circ \delta'|) \simeq (|\delta_1|, |\delta_2|)$. Therefore $|\delta'| \simeq (|\delta_1|, |\delta_2|)$, which is homotopic to the map constructed in the existence proof above. So $\delta$ is unique up to homotopy, meaning that $f$ is unique. This proves that $g_1$ and $g_2$ are satisfactory projection maps, proving that $T_1 \times T_2$ is in fact the categorical product of $T_1$ and $T_2$.

\end{proof}

The category $\textbf{pGrp}$ also has products. We define this product, and state without proof that it is indeed the categorical product. This follows immediately from the fact that the direct product of groups is the categorical product in $\textbf{Grp}$ \cite{MacLane}.

\begin{fact}

Let $(G_1, g_1)$ and $(G_2, g_2)$ be objects in $\emph{\textbf{pGrp}}$. Then $(G_1 \times G_2, (g_1, g_2))$ is their categorical product.

\end{fact}

Having defined the categorical products in $\textbf{Loop}$ and $\textbf{pGrp}$, and together with Theorem \ref{main}, the next corollary is a simple consequence.

\begin{cor}

\label{pp}

The functor $S : \emph{\textbf{Loop}} \rightarrow \emph{\textbf{pGrp}}$ preserves products.

\end{cor}

\begin{proof}

Let $T_1 = \Loop(\K_1, \lambda_1)$ and $T_2 = \Loop(\K_2, \lambda_2)$ be objects in \textbf{Loop}. Then $S(T_1) = (\pi_1(\K_1), [\lambda_1])$ and $S(T_2) = (\pi_1(\K_1), [\lambda_2])$, so $S(T_1) \times S(T_2) = (\pi_1(\K_1) \times \pi_2(\K_2), ([\lambda_1], [\lambda_2]))$. However, from the proof of Theorem \ref{topalg}, we see that $(\pi_1(\K_1) \times \pi_2(\K_2), ([\lambda_1], [\lambda_2])) \cong (\pi_1(\skel^2(\K_1 \times \K_2)), [\lambda_1 \star \lambda_2]) = S(T_1 \times T_2)$, so in fact $S(T_1 \times T_2) \cong S(T_1) \times S(T_2)$. Therefore $S$ preserves products.

\end{proof}

\section{Applications}

In this section we present some simple applications of the correspondence between compositions of loop agreement tasks and the products of their algebraic signatures. 

\begin{prop}

Let $T$ be $(3, 2)$-set agreement, and let $T'$ be any other loop agreement task. Then $T \times T'$ and $T$ are equivalent.

\end{prop}

\begin{proof}

Recall that $(3, 2)$-set agreement is the task $\Loop(\skel^1(\Delta^2), \zeta))$, where $\zeta$ is the triangle loop $(0, 1, 2, ((0, 1)), ((1, 2)), ((2, 0)))$. This triangle loop generates $\pi_1(\skel^1(\Delta^2))$, so $S(T) = (\pi_1(\skel^1(\Delta^2)), [\zeta]) \cong (\Z, 1)$. Let $S(T') = (G, g)$. Then by Corollary \ref{pp}, $S(T \times T') = S(T) \times S(T') = (\Z \times G, (1, g))$. The homomorphism $\phi : \Z \times G \rightarrow \Z$ defined by projection onto the first coordinate sends $(1, g)$ to $1$, and the homomorphism $\psi : \Z \rightarrow \Z \times G$ defined by $\psi(n) = (n, g)$ sends $1$ to $(1, g)$. So $T \times T'$ and $T$ implement one another, so are equivalent.

\end{proof}

Since $(3, 2)$-set agreement was shown to be universal for loop agreement by Herlihy and Rajsbaum \cite{Loop2}, it is operationally intuitive that composing it with any other loop agreement task should not change its relative power.

\begin{prop}

Let $T$ be any simplex agreement task, and let $T'$ be any other loop agreement task. Then $T \times T'$ and $T'$ are equivalent.

\end{prop}

\begin{proof}

Since the output complex if $T$ is a subdivided simplex, it has trivial fundamental group, so $S(T) = (1, e)$. As before, let $S(T') = (G, g)$. By Corollary \ref{pp}, $S(T \times T') = S(T) \times S(T') = (1 \times G, (g, e))$, which is clearly isomorphic to $(G, g)$. So $T \times T'$ and $T$ implement one another, so are equivalent.

\end{proof}

Herlihy and Rajsbaum also showed that simplex agreement is implemented from any loop agreement task \cite{Loop2}, so it is also intuitively clear that composing a task with simplex agreement should not change the relative power of the original task.

\begin{prop}

Let $T$ be any loop agreement task. Then $T \times T$ and $T$ are equivalent.

\end{prop}

\begin{proof}

Let $S(T) = (G, g)$. Then by Corollary \ref{pp}, $S(T \times T) = S(T) \times S(T) = (G \times G, (g, g))$. Letting $\phi : G \rightarrow G \times G$ be the diagonal map $\phi(x) = (x, x)$, $\phi$ maps $g$ to $(g, g)$, and letting $\psi : G \times G \rightarrow G$ be projection onto a coordinate, $\psi$ maps $(g, g)$ to $g$. So $T \times T$ and $T$ are equivalent.

\end{proof}

The above result states that composing a loop agreement task with copies of itself will not change its relative power.

\section{Conclusions}

It is a common technique to study a class of objects by mapping these objects into a class of simpler ones in such a way that preserves enough information about the original class of objects. This was the idea behind the fundamental group from algebraic topology, and was also the idea of the algebraic signature of Herlihy and Rajsbaum in their work on loop agreement. In this work we formalized and further extended the algebraic signature characterization by defining the composition of tasks and relating compositions of tasks to products of groups, and in doing so we partially answered the questions raised in the original paper. How much further can this characterization be extended; what more can we learn from the algebraic signature functor between loop agreement tasks and groups with distinguished elements? Does this functor have an adjunction?

The categorical techniques in this paper can be applied to general tasks. For example, tasks with decision maps form a category \textbf{Task}, with loop agreement as a subcategory. In the case of loop agreement, we are able to extract valuable information about tasks by mapping them into groups. What kind of functors may we apply to general tasks? Also in the case of loop agreement, we were able to identify parallel composition with the category product. Can parallel composition be defined for more general tasks, for instance via $\skel^n(\Oh_1 \times \Oh_2)$, and what is its precise operational meaning of parallel composition for general tasks?

\bibliography{Loop}
\bibliographystyle{plain}

\end{document}